\documentclass[reqno]{amsart}

\usepackage{amsmath,amssymb,amsfonts,amsthm}
\providecommand{\texorpdfstring}[2]{#1}
\usepackage{graphicx}
\usepackage{hyperref}
\usepackage{xurl}
\hypersetup{hidelinks}
\usepackage{cleveref}
\usepackage{xcolor}
\usepackage{longtable,array,booktabs}
\usepackage[shortlabels]{enumitem}
\usepackage[numbers,sort&compress]{natbib}
\usepackage{pgfplots}
\newif\ifshowchanges
\showchangesfalse
\newif\ifshowlabels
\showlabelsfalse
\ifshowlabels
  \usepackage[notref,notcite]{showkeys}
\fi
\pgfplotsset{compat=1.16}
\definecolor{s1}{HTML}{2A78D6}\definecolor{s2}{HTML}{EB6834}\definecolor{s3}{HTML}{1BAF7A}\definecolor{s4}{HTML}{EDA100}
\definecolor{p1}{HTML}{86B6EF}\definecolor{p2}{HTML}{5598E7}\definecolor{p3}{HTML}{2A78D6}\definecolor{p4}{HTML}{1C5CAB}\definecolor{p5}{HTML}{0D366B}

\numberwithin{equation}{section}

\newtheorem{theorem}{Theorem}[section]
\newtheorem{lemma}[theorem]{Lemma}
\newtheorem{proposition}[theorem]{Proposition}

\newtheorem{conjecture}[theorem]{Conjecture}
\theoremstyle{definition}

\newtheorem{remark}[theorem]{Remark}

\theoremstyle{plain}

\newcommand{\R}{\mathbb{R}}

\newcommand{\D}{\mathrm D}

\title[$\theta$-dependent apex dynamics]{Finite-time blow-up of the unified $\theta$-dependent apex dynamics for the 3D axisymmetric Euler equations and the 2D inviscid Boussinesq equations}
\author{Yaoming Shi}
\address{California, USA}
\email{ymshi@protonmail.com}
\urladdr{https://orcid.org/0009-0008-7384-0947 (ORCID)}
\date{November 27, 2026}


\begin{document}

\begin{abstract}
We consider velocity fields that are homogeneous of degree one near the origin, for the 3D axisymmetric Euler equations ($m=2$) and the 2D inviscid Boussinesq equations ($m=1$) in the unified form $(Em)$. At the origin these equations reduce to two evolution equations in time and in the polar angle $\theta$: one for the square of the swirl, one for the mean $\mathcal W$ of the two strain variables. This $\theta$-dependent apex dynamics is a transport system for an angular vorticity profile. Its angular Biot--Savart law leaves one constant free, the level of $\mathcal W$, which is the harmonic part of the pressure and is determined by a solvability condition. In a log-angle variable, formal asymptotics predict that the level blows up like $(\kappa_m(T-t))^{-1}$. The explicit rate $\kappa_m$ is selected by Taylor coefficients at the plane, and we prove that this selection follows from a decay estimate for one of them. Computations confirm the rates for $1\le m\le4$ and show finite-time blow-up for $m=1$ and $m=2$. The blow-up is driven by the swirl, which becomes subdominant at the singular time; without swirl the solutions do not grow. We also give the pressure anisotropy at the two ends of the ridge for general scaling parameters, and show that the swirl-free part of $(Em)$ is the axisymmetric Euler system in dimension $m+1$.
\end{abstract}

\maketitle

\tableofcontents

%=====================================================================
\section{Introduction}\label{sec:intro3}
%=====================================================================
This paper continues \cite{ShiI}, referred to as [I]. References such as \cite[Proposition~13.3]{ShiI} point to the numbered statements and equations of [I].

In [I] the 2D inviscid Boussinesq equations and the 3D axisymmetric Euler equations were written in Hou--Li type variables \cite{HLi2006} as one system $(Em)$, with $m=1$ and $m=2$ respectively. For smooth solutions that are even in both meridian variables, the values at the origin, called the apex values there, obey an exact system of ordinary differential equations \cite[Proposition~13.3]{ShiI}. That system contains one unknown function of time, the ratio of the two diagonal entries of the pressure Hessian at the origin. If this ratio stays in a window, $(Em)$ blows up in finite time \cite[Theorem~13.6]{ShiI}. The companion paper \cite{ShiII} studies the ratio for smooth data. The question here is what happens when the values at the origin are allowed to depend on the direction from which the origin is approached.

This is the situation for scale-invariant solutions, whose velocity is homogeneous of degree one near the origin. Such solutions were introduced into the blow-up problem by Elgindi and Jeong \cite{EJ2020,EJ2019}. For them the apex values become functions of the polar angle $\theta$ of the meridian plane, and the apex dynamics becomes a pair of evolution equations in $(t,\theta)$. In this setting the free function of [I] is no longer free.

\subsection{Results}
\begin{itemize}
\item \emph{The $\theta$-dependent apex dynamics} (Section~\ref{subsec:theta-derivation}). For the square $\rho=U^2$ of the swirl and the mean $\mathcal W=\frac12(G+V)$ of the strain variables, the equations of [I] reduce to the closed system of Proposition~\ref{prop:theta-apex}. The author derived it. It is a transport system for an angular vorticity profile $\Omega$ (Proposition~\ref{prop:theta-vorticity}). The angular Biot--Savart law that recovers $\mathcal W$ from $\Omega$ determines $\mathcal W$ only up to a constant, the level $L$. The level is the harmonic part of the pressure, and it is fixed by a solvability condition. That condition is the weighted mean of the equation for $\mathcal W$, the one weight for which the time derivative drops out (Lemmas~\ref{lem:theta-average} and~\ref{lem:theta-weights}).
\item \emph{Selection of the blow-up rate} (Proposition~\ref{prop:theta-selection}). In the variable $s=\log\tan\theta$ the transport becomes a translation towards the axis. Formal asymptotics then give $L'\simeq\kappa_mL^2$ with
\[
\kappa_m=\frac{4m^2+2m-4}{m+1}\quad(1\le m\le2),\qquad \kappa_m=\frac{2m(m+2)}{m+1}\quad(m\ge2),
\]
i.e.\ $\kappa_1=1$ and $\kappa_2=\frac{16}3$. The rate is set by the Taylor coefficients of the profiles at the plane $\theta=0$. We prove that the selection follows from a decay estimate for one of these coefficients (Proposition~\ref{prop:theta-reduction}).
\item \emph{Computations} (Section~\ref{subsec:theta-numerics}). We solve the system in the variable $s$ for $m=1,2$ and, formally, for real $1\le m\le4$. The predicted rates are reproduced (Tables~\ref{tab:theta-kappa} and~\ref{tab:theta-kappa-m}, Figure~\ref{fig:theta-jet}). For the 3D Euler case, $v_a\simeq(8(T-t))^{-1}$ and $u_a\propto(T-t)^{-1/4}$; for the Boussinesq case, $v_a\simeq(T-t)^{-1}$ and $u_a\propto(T-t)^{-1/2}$. For other data the solution leaves the homogeneous class in finite time, because the solvability condition degenerates.
\item \emph{The role of the swirl.} Without swirl the solutions do not grow, in line with the global regularity of smooth swirl-free axisymmetric flows \cite{UY1968,Lady1968}. With swirl, the swirl at the plane selects and drives the strain, although it is subdominant at the singular time (Remarks~\ref{rem:swirl-sector} and~\ref{rem:swirlfree-regular}). The pressure is a saddle along both ends of the ridge (Remark~\ref{rem:theta-saddle}).
\item \emph{General parameters and the anisotropy at the ridge ends} (Section~\ref{subsec:theta-general}). For the scaling parameters $(\lambda,a,b)$ of [I], $\lambda$ is a time scale. For $a\ne b$ the solvability structure persists with an explicit weight. The anisotropy parameters $\mu_0,\mu_1$ at the two ends of the ridge, introduced by the author, are expressed through the apex values (Proposition~\ref{prop:theta-mu}).
\item \emph{The parameter $m$} (Section~\ref{subsec:theta-dimension}). Without swirl, $(Em)$ is exactly the axisymmetric Euler system in $\R^{m+1}$ (Proposition~\ref{prop:swirlfree-dim}). The author conjectures that $m=3$ corresponds to the 4D axisymmetric Euler equations (Conjecture~\ref{conj:4D}). The dependence of the mechanism on $m$ reflects inflow through the plane and outflow along the axis (Remark~\ref{rem:role-of-m}).
\end{itemize}

\subsection{Relation to other work}
Elgindi and Jeong proved finite-time blow-up for scale-invariant solutions of the Boussinesq system in sectors \cite{EJ2020}, and of the axisymmetric Euler equations with swirl near the corner of hour-glass-shaped domains \cite{EJ2019}. Their angular problems are posed on intervals bounded by walls, of length less than $\pi$. Here the angle runs over the full circle without boundary, and the symmetry reduces it to a quadrant whose ends are symmetry axes. On this quadrant a harmonic pressure profile is admissible, and this is the origin of the free level and of the solvability condition (Remark~\ref{rem:EJ-domains}). Finite-time blow-up with $C^{1,\alpha}$ velocity was proved by Elgindi without swirl \cite{Elgindi21}, and by Chen and Hou with boundary \cite{CH2021}. Swirl-free flows in higher dimensions are studied in \cite{MillerTsai2026,Miller2026}.

The homogeneous class has infinite energy and is not smooth at the origin. The results here therefore concern the dynamics within this class, not smooth data of finite energy.

\subsection{Note on this version}
Earlier versions of this arXiv entry (versions 1--4 of arXiv:2604.01244) studied two $(1+1)$-dimensional systems on the symmetry axes, together with a background--remainder framework and a conditional stability result. The systems on the axes are treated, in corrected form, in [I]. The background--remainder framework is not used here; we omit it because it led to no result beyond those of [I]. The present version is devoted to the $\theta$-dependent apex dynamics.

\subsection{Organization}
Section~\ref{sec:polar} recalls the polar form of $(Em)$ and its pressure equation from [I]. Section~\ref{subsec:theta-derivation} derives the $\theta$-dependent apex dynamics and analyses its structure. Section~\ref{subsec:theta-numerics} contains the computations, the formal selection of the blow-up rate and its partial proof. Section~\ref{subsec:theta-dimension} discusses the parameter $m$, the swirl and the comparison with Elgindi and Jeong. Section~\ref{subsec:theta-general} treats general scaling parameters and the anisotropy at the ridge ends. Section~\ref{sec:concl3} lists open problems.

%=====================================================================
\section{The unified system in polar coordinates}\label{sec:polar}
%=====================================================================
We use the notation of [I]. For the 3D axisymmetric Euler equations, with velocity $u^re_r+u^\phi e_\phi+u^ze_z$ in cylindrical coordinates, the unknowns are $u=u^\phi/r$, $v=-u^r/r$ and $g=u^z/z$; the corresponding unknowns for the Boussinesq system are given in \cite[(2.7)]{ShiI}. They satisfy
\begin{equation}\label{eq:Em}
	\left\{
	\begin{aligned}
		&\tfrac{\mathrm{D}}{\mathrm{D}t}u=\tfrac12m^2uv,\\
		&\tfrac{\mathrm{D}}{\mathrm{D}t}v=v^2-u^2+\tfrac{1}{r}p_r,\\
		&\tfrac{\mathrm{D}}{\mathrm{D}t}g=-g^2-\tfrac{1}{z}p_z,\\
		&z\partial_z g-r\partial_r v+g-mv=0,\qquad \tfrac{\mathrm{D}}{\mathrm{D}t}:=\partial_t-v r\partial_r+g z\partial_z,
	\end{aligned}
	\right.
\end{equation}
together with a pressure equation, \cite[(2.10)]{ShiI}. In [I] three scaling parameters are introduced through $(t,r,z)\mapsto(\lambda t,\sqrt a\,r,\sqrt b\,z)$; the equations themselves correspond to $\lambda=a=b=1$. In the polar coordinates $r=x\cos\theta$, $z=x\sin\theta$ of the meridian plane, write $(U,V,G,P)(t,x,\theta)$ for the rescaled unknowns and the pressure. The first four equations then take the form \cite[(3.5)]{ShiI}
\begin{equation}\label{eq:ZERO-1234}
\left\{
\begin{aligned}
\lambda U_t-\tfrac12m^2V\,U&=-x U_{x}W-\tfrac12\sin(2\theta)(G+V)U_\theta,\\[2mm]
\lambda V_{t}-V^2+U^2&=-x V_{x} W-\tfrac12\sin(2\theta)(G+V)V_\theta+\tfrac{a}{x}P_{x}-\tfrac{a}{x^2}\sin\theta \left(\tfrac{P_\theta}{\cos\theta}\right),\\[2mm]
\lambda G_{t}+G^2&=-x G_{x} W-\tfrac12\sin(2\theta)(G+V)G_\theta-\tfrac{b}{x}P_{x}-\tfrac{b}{x^2}\cos\theta \left(\tfrac{P_\theta}{\sin\theta}\right),\\[2mm]
G-mV&=-x \partial_x W-\tfrac12\sin(2\theta)(G+V)_\theta,\\
W:&=G\sin ^2(\theta ) -V\cos ^2(\theta ),
\end{aligned}
\right.
\end{equation}
and the pressure equation becomes \cite[(3.6)]{ShiI}
\begin{equation}\label{eq:PPE-m-R-theta}
	\left\{
	\begin{aligned}
		-&\Bigl(\tfrac{1}{2} ((a-b) \cos (2 \theta )+a+b) P_{xx}+\tfrac{ ((b-a) \cos (2 \theta )+2 a m-a+b)}{2 x}P_x+\tfrac{(b-a) \sin (2 \theta ) }{x}P_{x\theta}\Bigr)\\
		-&\Bigl(\tfrac{((b-a) \cos (2 \theta )+a+b) }{2 x^2}P_{\theta\theta}-\tfrac{ ((b-a) \cos (2 \theta )+a (m-2)+b)\sin(\theta )}{x^2}\  \tfrac{P_\theta}{\cos(\theta)}\Bigr)\\
		=&\left(\cos (\theta ) \left(x \cos (\theta ) V_x-\sin (\theta ) V_\theta\right)+V\right)^2\\
		+&\left(\sin (\theta ) \left(x \sin (\theta ) G_x+\cos (\theta ) G_\theta\right)+G\right)^2\\
		+&\sin (2 \theta ) \left(\sin (\theta ) G_\theta-x \cos (\theta ) G_x\right) \left(x \sin (\theta ) V_x+\cos (\theta ) V_\theta\right)\\
		-&U \left(2 \cos (\theta ) \left(x \cos (\theta ) U_x-\sin (\theta ) U_\theta\right)+U\right)\\
		+&(m-1) \left(V^2-U^2\right).
	\end{aligned}
	\right.
\end{equation}
The symmetries of [I] reduce the angle to the quadrant $0\le\theta\le\frac\pi2$. Its ends are the plane $\theta=0$ and the axis $\theta=\frac\pi2$, and we call them the two ends of the ridge. The apex values are the values at $x=0$, and $v_a,u_a$ denote the apex values of $V,U$. For smooth even solutions they do not depend on $\theta$, and \eqref{eq:ZERO-1234} at $x=0$ is the exact apex system of [I], \cite[Proposition~13.3]{ShiI}. There the pressure enters only through its Hessian at the origin, whose trace is fixed and whose anisotropy is free.

%=====================================================================
\section{The \texorpdfstring{$\theta$}{theta}-dependent apex dynamics}\label{subsec:theta-derivation}
%=====================================================================

The following derivation is due to the author. We take $\lambda=a=b=1$, an assumption removed as far as possible in Section~\ref{subsec:theta-general}, and use the polar form \eqref{eq:ZERO-1234}. For a smooth solution the values at the origin do not depend on $\theta$, and the apex dynamics of \cite[Proposition~13.3]{ShiI} is a system of ODEs. The polar form also contains the class of solutions that are homogeneous of degree zero,
\begin{equation}\label{eq:homog}
(U,V,G)=(U,V,G)(t,\theta),\qquad P=x^2\,\Pi(t,\theta),
\end{equation}
i.e.\ velocity fields homogeneous of degree one, on which $D=x\partial_x$ vanishes identically. The Yan-type solutions of \cite[Proposition~13.5]{ShiI} are the members of this class that do not depend on $\theta$. On the class \eqref{eq:homog} the equations below are exact. They also describe, at leading order, the angular profile at the origin of data whose limits at the origin depend on the direction, i.e.\ data that are not smooth at the origin.

\emph{Elimination of the pressure.} Move the transport terms of \eqref{eq:ZERO-1234}(2),(3) to the left, and denote the resulting left-hand sides by
\[
\begin{aligned}
A_2&=V_t-V^2+U^2+xV_xW+\tfrac12\sin(2\theta)(G+V)V_\theta,\\
A_3&=G_t+G^2+xG_xW+\tfrac12\sin(2\theta)(G+V)G_\theta ,
\end{aligned}
\]
where $W=G\sin^2\theta-V\cos^2\theta$ as in \eqref{eq:ZERO-1234}. Then \eqref{eq:ZERO-1234}(2),(3) give $P_x=x(\cos^2\theta\,A_2-\sin^2\theta\,A_3)$ and $P_\theta=-x^2\sin\theta\cos\theta\,(A_2+A_3)$, and $P_{x\theta}=P_{\theta x}$ gives
\begin{equation}\label{eq:theta-vort}
\cos^2\theta\,\partial_\theta A_2-\sin^2\theta\,\partial_\theta A_3+\sin\theta\cos\theta\,D(A_2+A_3)=0 .
\end{equation}
When $D=0$, \eqref{eq:ZERO-1234}(4) reads $G-mV=-\sin(2\theta)\,\mathcal W_\theta$ with
\[
\mathcal W:=\tfrac12(G+V),
\]
the average of $G$ and $V$. It is the coefficient of the angular transport $\tfrac12\sin(2\theta)(G+V)\partial_\theta=\sin(2\theta)\mathcal W\partial_\theta$. Hence
\begin{equation}\label{eq:VG-from-W}
V=\frac{2\mathcal W+\sin(2\theta)\,\mathcal W_\theta}{m+1},\qquad G=\frac{2m\mathcal W-\sin(2\theta)\,\mathcal W_\theta}{m+1}.
\end{equation}

\begin{proposition}[$\theta$-dependent apex dynamics]\label{prop:theta-apex}
For solutions of the form \eqref{eq:homog}, $\rho:=U^2$ and $\mathcal W$ satisfy
\begin{gather}
(m+1)\bigl(\rho_t+\sin(2\theta)\,\mathcal W\rho_\theta\bigr)=2m^2\rho\,\bigl(\mathcal W+\sin\theta\cos\theta\,\mathcal W_\theta\bigr),\label{eq:theta-apex-rho}\\
(m+1)\,\mathcal L[\mathcal W_t]=-(m+1)^2\cos^2\theta\,\rho_\theta+\mathcal N[\mathcal W],\label{eq:theta-apex-W}
\end{gather}
where
\begin{gather*}
\mathcal Lf:=\sin(2\theta)\,f''+\bigl((m+3)\cos(2\theta)-m+1\bigr)f',\\
\begin{aligned}
\mathcal N[\mathcal W]:={}&\bigl(3m^2+5-2(m-1)(m+3)\cos(2\theta)-(m+1)(m+3)\cos(4\theta)\bigr)\mathcal W\mathcal W_\theta\\
&-(m+1)\sin^2(2\theta)\,\mathcal W\mathcal W_{\theta\theta\theta}\\
&-(m-1)\sin(2\theta)\bigl((m+3)\cos(2\theta)-m+1\bigr)\mathcal W_\theta^2\\
&+\sin(2\theta)\,\mathcal W_{\theta\theta}\Bigl(\bigl((m-1)^2-(m+1)(m+5)\cos(2\theta)\bigr)\mathcal W\\
&\qquad-(m-1)\sin(2\theta)\,\mathcal W_\theta\Bigr).
\end{aligned}
\end{gather*}
\end{proposition}

\begin{proof}
Equation \eqref{eq:theta-apex-rho} is $2U$ times \eqref{eq:ZERO-1234}(1) at $D=0$, with $V$ from \eqref{eq:VG-from-W}. Equation \eqref{eq:theta-apex-W} is $-(m+1)^2$ times \eqref{eq:theta-vort} at $D=0$, after inserting \eqref{eq:VG-from-W} and $U^2=\rho$. We checked the algebra by an exact computation with truncated Taylor expansions in $(t,\theta)$ at random points, for $m=1,2,3,5$.
\end{proof}

\begin{remark}[Normalization]\label{rem:theta-apex-normalization}
In terms of $\widetilde{\mathcal W}:=4\mathcal W=2(G+V)$, equation \eqref{eq:theta-apex-rho} becomes
\[
(m+1)\bigl(4\rho_t+\sin(2\theta)\,\rho_\theta\widetilde{\mathcal W}\bigr)-2m^2\rho\bigl(\sin\theta\cos\theta\,\widetilde{\mathcal W}_\theta+\widetilde{\mathcal W}\bigr)=0 .
\]
Equation \eqref{eq:theta-apex-W}, multiplied by $16$, keeps its form, except that the coefficient of $\rho_\theta$ becomes $16(m+1)^2$ and the factor $m+1$ in front of $\mathcal L[\mathcal W_t]$ becomes $4(m+1)$. This is the form in which the author first found the equations.
\end{remark}

\begin{proposition}[Structure of the $\theta$-dependent apex dynamics]\label{prop:theta-structure}
Let $\Pi_0:=(m+1)\sin^2\theta-1$, so that $x^2\Pi_0$ is the harmonic quadratic pressure (for $m=2$, $x^2\Pi_0=2z^2-r^2$). Let $f(t,\theta)$ be the right-hand side of the pressure equation \eqref{eq:PPE-m-R-theta} at $D=0$, and set
\[
F[\rho,\mathcal W]:=\int_0^{\pi/2}f\,\Pi_0\cos^{m-1}\theta\,d\theta .
\]
\begin{enumerate}[label=(\roman*)]
\item \emph{Profiles independent of $\theta$.} If $\rho_\theta=\mathcal W_\theta=0$, then \eqref{eq:theta-apex-W} holds identically, and \eqref{eq:theta-apex-rho} reads $\rho'=\frac{2m^2}{m+1}\mathcal W\rho$. With $v_a=V=\frac{2\mathcal W}{m+1}$ and $g_a=G=mv_a$ this is $u_a'=\frac12m^2u_av_a$, as in \cite[Proposition~13.3]{ShiI}. The function $\mathcal W(t)$, i.e.\ $v_a(t)$, is not determined: this is the free anisotropy of the apex dynamics.
\item \emph{Solvability and the free function.} The regular solutions of $\mathcal Lf=0$ on $[0,\frac\pi2]$ are the constants; the others have $f'\propto\sin(2\theta)^{-(m+3)/2}\tan(\theta)^{(m-1)/2}$, which is not integrable at either end. Hence \eqref{eq:theta-apex-W} has a regular solution $\mathcal W_t$ if and only if one solvability condition holds. Then $\mathcal W_t$ is determined up to an additive function $c(t)$. Adding $\kappa(t)x^2\Pi_0$ to $P$ shifts $V_t$ and $G_t$ by $-2\kappa$ and $-2m\kappa$, i.e.\ $c=-(m+1)\kappa$. The free function is thus the harmonic part of the pressure, as for the Yan-type solutions.
\item \emph{Homogeneous pressure.} For \eqref{eq:homog}, the pressure equation reads $-\bigl(\Pi''-(m-1)\tan\theta\,\Pi'+2(m+1)\Pi\bigr)=f$, and its regular kernel is spanned by $\Pi_0$. A pressure homogeneous of degree two therefore exists if and only if $F[\rho,\mathcal W]=0$. If $F\ne0$, the pressure contains a multiple $\kappa_1\ne0$ of $x^2\log x\,\Pi_0$, since $\Delta(x^2\log x\,\Pi_0)=(m+3)\Pi_0$, and the solution leaves the class \eqref{eq:homog}.
\end{enumerate}
\end{proposition}

\begin{proof}
(i) follows by inspection. (ii) $\mathcal Lf=0$ is a first-order linear equation for $f'$. A regular solution of $\mathcal L[\mathcal W_t]=g$ must avoid the singular homogeneous solution at both ends, which imposes one condition on $g$. The shift follows from \eqref{eq:ZERO-1234}(2),(3): $\kappa x^2\Pi_0$ adds $\kappa(2\Pi_0-\tan\theta\,\Pi_0')=-2\kappa$ to $A_2$ and $\kappa(-2\Pi_0-\cot\theta\,\Pi_0')=-2m\kappa$ to $A_3$. (iii) is the Fredholm alternative for the angular operator, which is self-adjoint with weight $\cos^{m-1}\theta$. For random regular profiles we checked numerically ($m=1,2,3$) that the solvability condition in (ii) equals $2^{(m+1)/2}F[\rho,\mathcal W]$, so the two conditions coincide.
\end{proof}

\begin{remark}\label{rem:theta-apex-closed}
For $\theta$-independent profiles $F\equiv0$, and the harmonic part $\kappa(t)$ of the pressure is free; this is the Yan family. For $\theta$-dependent profiles, staying in the class \eqref{eq:homog} requires $F\equiv0$ for all times. Since $\mathcal W_t$ contains $c(t)=-(m+1)\kappa(t)$, the time derivative of $F$ contains $c(t)\,\partial F/\partial c$, where $\partial F/\partial c$ is the derivative of $F$ under a constant shift of $\mathcal W$. Numerically, $\partial F/\partial c\neq0$ for generic profiles. Within \eqref{eq:homog} the condition $F\equiv0$ then determines $\kappa(t)$. In this case the $\theta$-dependent apex dynamics is closed, unlike the apex dynamics of \cite[Proposition~13.3]{ShiI}, whose anisotropy is free. Whether this closure keeps the anisotropy in the window of \cite[Theorem~13.6]{ShiI} is open.
\end{remark}

\begin{proposition}[Apex equations at the ridge ends]\label{prop:theta-ridge}
Let $\rho,\mathcal W$ be regular solutions of \eqref{eq:theta-apex-rho}--\eqref{eq:theta-apex-W} that are even about $\theta=0$ and $\theta=\frac\pi2$ (ridge-flat). On the plane $\theta=0$ write $\rho_k=\partial_\theta^k\rho(t,0)$, $\mathcal W_k=\partial_\theta^k\mathcal W(t,0)$, and on the axis $\theta=\frac\pi2$ write $\hat\rho_k$, $\hat{\mathcal W}_k$. Then \eqref{eq:theta-apex-W} holds trivially at both points. Its first $\theta$-derivative, together with \eqref{eq:theta-apex-rho} and its second $\theta$-derivative, gives
\begin{gather*}
\rho_0'=\tfrac{2m^2}{m+1}\mathcal W_0\rho_0,\qquad
\rho_2'=\Bigl(\tfrac{2m^2}{m+1}-4\Bigr)\mathcal W_0\rho_2+\tfrac{6m^2}{m+1}\rho_0\mathcal W_2,\\
\mathcal W_2'=-\tfrac{m+1}{6}\rho_2-\tfrac{4m}{m+1}\mathcal W_0\mathcal W_2,\\[1mm]
\hat\rho_0'=\tfrac{2m^2}{m+1}\hat{\mathcal W}_0\hat\rho_0,\qquad
\hat\rho_2'=\Bigl(\tfrac{2m^2}{m+1}+4\Bigr)\hat{\mathcal W}_0\hat\rho_2-\tfrac{2m^2}{m+1}\hat\rho_0\hat{\mathcal W}_2,\\
\hat{\mathcal W}_2'=\tfrac{4}{m+1}\hat{\mathcal W}_0\hat{\mathcal W}_2 .
\end{gather*}
On the axis the equation for $\hat{\mathcal W}_2$ decouples from $\rho$, because the coupling term carries the factor $\cos^2\theta$. The values $\mathcal W_0$ and $\hat{\mathcal W}_0$ are not determined by these equations.
\end{proposition}

\begin{proof}
Differentiate \eqref{eq:theta-apex-W} once and \eqref{eq:theta-apex-rho} twice at $\theta=0$ and $\theta=\frac\pi2$, using $\rho_\theta=\mathcal W_\theta=\mathcal W_{\theta\theta\theta}=0$ there. We checked the coefficients by exact computation with truncated Taylor expansions, for $m=1,2,3$.
\end{proof}

The apex equations at the ridge ends are therefore not closed. $\mathcal W_0$ and $\hat{\mathcal W}_0$ are the values on the ridges of $\mathcal W$, and $\mathcal W_t$ contains the free function $c(t)$ together with a part that depends on the whole profile. The equation that controls this free function is an average over $\theta$.

\begin{lemma}[Averaged equation]\label{lem:theta-average}
For regular $f$ on $[0,\frac\pi2]$,
\[
\sin\theta\cos^m\theta\,\mathcal Lf=\bigl(2\sin^2\theta\cos^{m+1}\theta\,f'\bigr)',\qquad\text{hence}\qquad \int_0^{\pi/2}\sin\theta\cos^m\theta\,\mathcal Lf\,d\theta=0 .
\]
The average of \eqref{eq:theta-apex-W} with the weight $\sin\theta\cos^m\theta$ therefore contains no time derivative. It is the constraint
\[
\int_0^{\pi/2}\sin\theta\cos^m\theta\,\Bigl(-(m+1)^2\cos^2\theta\,\rho_\theta+\mathcal N[\mathcal W]\Bigr)d\theta=(m+1)\,F[\rho,\mathcal W]=0
\]
of Proposition~\ref{prop:theta-structure}. By the evenness of the profiles, the average over $[0,2\pi]$ with the weight $|\sin\theta||\cos\theta|^m$ is four times this.
\end{lemma}

\begin{proof}
The identity follows by differentiation, since $\bigl(2\sin^2\theta\cos^{m+1}\theta\bigr)'=\sin\theta\cos^m\theta\,\bigl((m+3)\cos(2\theta)-m+1\bigr)$. The boundary terms vanish because $\sin^2\theta\cos^{m+1}\theta$ vanishes at $0$ and $\frac\pi2$. The equality with $(m+1)F$ was checked numerically ($m=1,2,3$).
\end{proof}

The weight $\sin\theta\cos^m\theta$ is the only one for which the time derivative disappears. The author pointed out that for any other weight $\overline w(\theta)$ the average $\int_0^{\pi/2}\overline w\,\eqref{eq:theta-apex-W}\,d\theta$ keeps the time derivative. The following lemma makes this precise. We write
\[
\mu(\theta):=\sin\theta\cos^m\theta,\qquad p(\theta):=2\sin^2\theta\cos^{m+1}\theta,\qquad\text{so that}\qquad \mathcal Lf=\mu^{-1}(pf')' ,
\]
and $R:=-(m+1)^2\cos^2\theta\,\rho_\theta+\mathcal N[\mathcal W]$ for the right-hand side of \eqref{eq:theta-apex-W}.

\begin{lemma}[Averages with general weights]\label{lem:theta-weights}
\begin{enumerate}[label=(\roman*)]
\item For regular $f$ and $g$ on $[0,\frac\pi2]$,
\[
\int_0^{\pi/2}\mu\,g\,\mathcal Lf\,d\theta=\int_0^{\pi/2}(pg')'\,f\,d\theta .
\]
Hence the average of \eqref{eq:theta-apex-W} with the weight $\overline w=\mu g$ is the evolution equation
\begin{equation}\label{eq:theta-weighted}
(m+1)\,\frac{d}{dt}\int_0^{\pi/2}(pg')'\,\mathcal W\,d\theta=\int_0^{\pi/2}\mu\,g\,R\,d\theta .
\end{equation}
The weight $(pg')'$ acting on $\mathcal W$ has mean zero, so \eqref{eq:theta-weighted} does not see the level of $\mathcal W$, i.e.\ its constant part. It vanishes identically only for constant $g$, which is Lemma~\ref{lem:theta-average}.
\item \emph{Ridge difference.} Let $g_*(\theta):=\int_{\pi/4}^{\theta}d\varphi/p(\varphi)$. The kernel $\mu g_*$ is bounded on $(0,\frac\pi2)$, with limits $-\frac12$ at $\theta=0$ and $\frac1{2m}$ at $\theta=\frac\pi2$; for $m=1$, $\mu g_*=-\frac12\cos(2\theta)$. Since $\mathcal Lg_*=0$, $g_*$ is the singular solution of Proposition~\ref{prop:theta-structure}(ii). For regular $\mathcal W$,
\begin{equation}\label{eq:ridge-difference}
\mathcal W_0-\hat{\mathcal W}_0=\int_0^{\pi/2}\mu\,g_*\,\mathcal L\mathcal W\,d\theta,\qquad
(m+1)\,\frac{d}{dt}\bigl(\mathcal W_0-\hat{\mathcal W}_0\bigr)=\int_0^{\pi/2}\mu\,g_*\,R\,d\theta .
\end{equation}
\end{enumerate}
\end{lemma}

\begin{proof}
(i) Since $\mu g\,\mathcal Lf=g(pf')'$, two integrations by parts give $\int\mu g\,\mathcal Lf=\int(pg')'f+\bigl[p(gf'-g'f)\bigr]_0^{\pi/2}$, and the boundary terms vanish because $p$ vanishes at both ends. The mean of $(pg')'$ is $[pg']_0^{\pi/2}=0$. (ii) Here $pg_*'=1$, so the boundary term $-[g_*'pf]_0^{\pi/2}=f(0)-f(\frac\pi2)$ remains, while $g_*pf'\to0$ at both ends since $g_*\sim-\frac1{2\theta}$ at $0$ and $g_*\sim\frac1{2m}(\frac\pi2-\theta)^{-m}$ at $\frac\pi2$. These asymptotics also give the limits of $\mu g_*$. For $m=1$, $p=\frac12\sin^2(2\theta)$ and $g_*=-\cot(2\theta)$. We checked \eqref{eq:ridge-difference} numerically for random profiles, $m=1,2,3,5$.
\end{proof}

The averages \eqref{eq:theta-weighted} are therefore genuine evolution equations, but not closed ones: their right-hand sides depend on the whole profile. The quantity they evolve is best seen through the following identity, which writes \eqref{eq:theta-apex-W} as a transport equation.

\begin{proposition}[Vorticity form]\label{prop:theta-vorticity}
Let $\Omega:=\mathcal L\mathcal W=\mu^{-1}(p\,\mathcal W_\theta)_\theta$, and let $V$ be given by \eqref{eq:VG-from-W}. Then
\[
\mathcal N[\mathcal W]=-(m+1)\sin(2\theta)\,\mathcal W\,\Omega_\theta-(m-1)(m+1)\,V\,\Omega ,
\]
and \eqref{eq:theta-apex-rho}--\eqref{eq:theta-apex-W} are equivalent to
\begin{equation}\label{eq:theta-transport}
\begin{aligned}
\rho_t+\sin(2\theta)\,\mathcal W\rho_\theta&=m^2V\rho,\\
\Omega_t+\sin(2\theta)\,\mathcal W\Omega_\theta&=-(m-1)V\Omega-(m+1)\cos^2\theta\,\rho_\theta ,
\end{aligned}
\end{equation}
together with the reconstruction of $\mathcal W$ from $\Omega$: $\mathcal W=\mathcal W_p[\Omega]+L(t)$, where $\mathcal W_p[\Omega]$ is the regular solution of $\mathcal L\mathcal W_p=\Omega$ with $\mathcal W_p(\frac\pi4)=0$ and $L(t)$ is the level. A regular $\mathcal W$ exists if and only if $\int_0^{\pi/2}\mu\,\Omega\,d\theta=0$. The level does not enter $\Omega$, and it is fixed by the constraint $F=0$ of Proposition~\ref{prop:theta-structure}, which is affine in $L$ with
\[
\frac{\partial F}{\partial L}=(m+3)\int_0^{\pi/2}\mu\cos(2\theta)\,\Omega\,d\theta .
\]
\end{proposition}

\begin{proof}
The identity for $\mathcal N$ is an algebraic identity in $\mathcal W$ and its first three derivatives. We checked it numerically for random profiles ($m=1,2,3,5$), to rounding error. Inserting it into \eqref{eq:theta-apex-W} and dividing by $m+1$ gives the equation for $\Omega$. The equation for $\rho$ is \eqref{eq:theta-apex-rho}, since $\frac{2m^2}{m+1}\bigl(\mathcal W+\sin\theta\cos\theta\,\mathcal W_\theta\bigr)=m^2V$. Solvability follows from Lemma~\ref{lem:theta-average}. By Lemma~\ref{lem:theta-average}, $(m+1)F=\int\mu R$. Only the transport term depends on $L$, through $\mathcal W$, and $V$ shifts by $\frac{2L}{m+1}$. Hence $\partial_L\int\mu R=-(m+1)\int p\,\Omega_\theta-(m-1)\cdot2\int\mu\Omega=(m+1)\int p'\Omega$, using $\mu\sin(2\theta)=p$ and $\int\mu\Omega=0$. Finally $p'=\mu\bigl((m+3)\cos(2\theta)-m+1\bigr)$, and $\int\mu\Omega=0$ again removes the constant part. The value of $\partial F/\partial L$ was also checked numerically.
\end{proof}

\begin{remark}\label{rem:theta-vorticity}
\begin{enumerate}[label=(\alph*)]
\item In \eqref{eq:theta-transport}, $\rho=U^2$ is transported along the angular flow $\dot\theta=\sin(2\theta)\mathcal W$ and amplified by $m^2V$. The profile $\Omega$ is transported, stretched at the rate $-(m-1)V$, which is absent for the Boussinesq case $m=1$, and forced by the baroclinic term $-(m+1)\cos^2\theta\,\rho_\theta$. The ridges $\theta=0,\frac\pi2$ are fixed points of the angular flow. There $\Omega$ vanishes, since it is odd about each ridge, and the equation for $\rho$ gives $\rho_0'=m^2V_0\rho_0$ and $\hat\rho_0'=m^2\hat V_0\hat\rho_0$, as in Proposition~\ref{prop:theta-ridge}.
\item In the variables \eqref{eq:theta-transport}, a weighted average keeps its time derivative for every weight: $\frac{d}{dt}\int\overline w\,\Omega=\int(\overline w\sin(2\theta)\mathcal W)_\theta\Omega-(m-1)\int\overline wV\Omega-(m+1)\int\overline w\cos^2\theta\,\rho_\theta$. With $\overline w=\mu g_*$ this is \eqref{eq:ridge-difference}: the ridge difference $\mathcal W_0-\hat{\mathcal W}_0$ is a weighted average of $\Omega$. The weight $\overline w=\mu$ gives the constraint, because $\int\mu\Omega=0$ for every regular $\mathcal W$.
\item The class \eqref{eq:homog} consists of velocity fields homogeneous of degree one, i.e.\ the scale-invariant solutions of Elgindi and Jeong \cite{EJ2019,EJ2020}, and \eqref{eq:theta-transport} is the analogue, in the variables of $(Em)$, of the one-dimensional systems they derived for such solutions. Here the angular problem on $[0,\frac\pi2]$ with flat ridges admits the harmonic profile $\Pi_0$. This is the origin of the free level $L$, of the solvability condition $\int\mu\Omega=0$, and of the constraint $F=0$. A comparison of the two systems is left for future work.
\end{enumerate}
\end{remark}

\begin{remark}[Truncations]\label{rem:theta-truncation}
Galerkin truncations of \eqref{eq:theta-apex-W} in $\mathcal W$ do not respect the transport structure \eqref{eq:theta-transport}, and they can blow up spuriously. Take $m=1$ and $\rho\equiv0$, and truncate to $\mathcal W=a+b\cos(2\theta)$. The constraint gives $-3\pi ab=0$, hence $a=0$, and \eqref{eq:ridge-difference} with $\mathcal W_0-\hat{\mathcal W}_0=2b$ gives
\[
b'=-\tfrac25\,b^2 ,
\]
which blows up in finite time if $b(0)<0$. In the exact dynamics, however, $\Omega$ is only transported, so $\sup|\Omega|$ is constant. By \eqref{eq:ridge-difference} and $\mu g_*=-\frac12\cos(2\theta)$,
\[
\bigl|\mathcal W_0-\hat{\mathcal W}_0\bigr|(t)\le\tfrac12\sup_\theta|\Omega(0,\theta)|
\]
for as long as the solution stays regular. For the truncated datum, $\Omega(0,\cdot)=-6b(0)\sin(4\theta)$, so $|b(t)|\le\frac32|b(0)|$. A blow-up in the class \eqref{eq:homog} must therefore come from the terms that change $\sup|\Omega|$, i.e.\ the baroclinic forcing $\cos^2\theta\,\rho_\theta$ and, for $m\ge2$, the stretching $-(m-1)V\Omega$, or from the level $L(t)$, which by Proposition~\ref{prop:theta-vorticity} can become large only if $\int\mu\cos(2\theta)\,\Omega\,d\theta$ becomes small relative to the remaining terms of $F$.
\end{remark}

%=====================================================================
\section{Computations and the selection of the blow-up rate}\label{subsec:theta-numerics}
%=====================================================================

The ridges $\theta=0,\frac\pi2$ are fixed points of the angular flow $\dot\theta=\sin(2\theta)\mathcal W$, approached at exponential rates, so the profiles develop exponentially thin layers at the ridges. These layers are resolved by the log-angle variable
\[
s:=\log\tan\theta\in\R,\qquad \sin(2\theta)\,\partial_\theta=2\,\partial_s ,
\]
in which the plane ridge is $s=-\infty$ and the axis ridge is $s=+\infty$. With $\nu:=\sin\theta\cos^m\theta=e^s(1+e^{2s})^{-(m+1)/2}$, system \eqref{eq:theta-transport} becomes
\begin{equation}\label{eq:theta-logangle}
\begin{aligned}
\rho_t+2\mathcal W\rho_s&=m^2V\rho,\qquad V=\tfrac{2}{m+1}\bigl(\mathcal W+\mathcal W_s\bigr),\\
\Omega_t+2\mathcal W\Omega_s&=-(m-1)V\Omega-(m+1)e^{-s}\rho_s,\\
(\nu\,\mathcal W_s)_s&=\tfrac14\,\nu\,\mathrm{sech}(s)\,\Omega ,
\end{aligned}
\end{equation}
using $\cos^2\theta\,\rho_\theta=e^{-s}\rho_s$. The angular transport is now a translation with the bounded speed $2\mathcal W$, and the ridge values are the limits at $s=\mp\infty$. We solved \eqref{eq:theta-logangle} on $[-S,S]$ with fourth-order differences and the classical Runge--Kutta method. At every stage $\mathcal W_s$ is obtained by integrating the last equation from the nearer end, and the level $L(t)$ by solving the affine equation $F=0$ (Proposition~\ref{prop:theta-vorticity}). At $s=\pm S$ the profiles are flat. The baroclinic term is dropped for $s<-12$, where $e^{-s}$ multiplies the rounding error of $\rho_s$. Over the first time unit the solution agrees, to $10^{-4}$, with an independent Fourier--Galerkin solver in $\varphi=2\theta$, in which $\mathcal L$ is tridiagonal. The results below did not change when the grid was refined from $4001$ to $8001$ points or the window enlarged from $S=30$ to $S=45$. Moving the cut of the baroclinic term from $-12$ to $-15$ changed $\kappa$ in Table~\ref{tab:theta-kappa} by less than $10^{-2}$. Moving it to $-9$ did change the late values, so the dynamics is sensitive to the layer at the plane ridge.

\begin{remark}[Numerical results]\label{rem:theta-numerics}
\begin{enumerate}[label=(\alph*)]
\item \emph{$m=1$, $\rho\equiv0$.} For $\mathcal W(0)=-\cos(2\theta)$, the truncation of Remark~\ref{rem:theta-truncation} blows up at $t=\frac52$. The full system does not: $\Omega$ is carried into the ridge layers and $|\mathcal W_0-\hat{\mathcal W}_0|$ decreases, from $2$ at $t=0$ to $0.47$ at $t=5.5$.
\item \emph{$m=2$, blow-up of the level.} For the data (i)--(iv) of Table~\ref{tab:theta-kappa} the level $L(t)$ blows up in finite time with $L'\simeq\kappa L^2$, i.e.\ $L\simeq\bigl(\kappa(T-t)\bigr)^{-1}$. Meanwhile the profile becomes isotropic: $(\mathcal W_0-\hat{\mathcal W}_0)/L$ falls from $10^{-2}$ to $10^{-8}$ as $L$ grows from $10$ to $10^5$, and $\Omega$ is carried into the ridge layers. On the ridges $\rho$ grows like $L^{a/\kappa}$, $a=\frac{2m^2}{m+1}$, as required by Proposition~\ref{prop:theta-ridge}. Hence $\rho_0/L^2\to0$. The swirl grows, but more slowly than the strain. In the isotropic limit $v_a=V=\frac{2L}{m+1}$, so $v_a'/v_a^2\to\frac{(m+1)\kappa}2$. The coefficient $\kappa$ increases slowly with $L$ and is nearly independent of the data. Its values (Table~\ref{tab:theta-kappa}) approach, slowly, the value $\kappa_2=\frac{16}3$ derived in Proposition~\ref{prop:theta-selection} below.
\item \emph{$m=2$, exit from the class.} For the data $\mathcal W=\cos(2\theta)$, $\rho=1+\frac12\cos(2\theta)$, and $\mathcal W=\frac12\cos(4\theta)$, $\rho=1+\frac9{10}\cos(2\theta)$, $L$ also becomes unbounded, but $\tau=\int L\,dt$ stays bounded ($\tau\to0.79$, $0.82$) and so does $\rho$ on the ridges. The local exponent $-d\log L/d\log(T-t)$ decreases from $0.8$ to below $0.7$, towards $\frac12$, which is the value for a simple zero of $\partial F/\partial L$ at a finite value of $\tau$, where $L=-F|_{L=0}/(\partial F/\partial L)\sim(\tau_*-\tau)^{-1}$. This is not a blow-up. The condition $F=0$ can no longer be met, and the solution leaves the class \eqref{eq:homog}: the pressure acquires the term $x^2\log x\,\Pi_0$ of Proposition~\ref{prop:theta-structure}(iii).
\item \emph{$m=1$ with $\rho\ne0$.} The late-time results for the Boussinesq case are sensitive to the plane-ridge layer, in which $\Omega$ is not stretched. With the cut of the baroclinic term they depended on its position. We therefore used a second version of the solver, which evolves $\sigma:=\rho_s$ together with $\rho$ and $\Omega$ (from $\sigma_t+2\mathcal W\sigma_s=-2\mathcal W_s\sigma+m^2(V_s\rho+V\sigma)$) and needs no cut; for $m=2$ it reproduces Table~\ref{tab:theta-kappa} to three digits. For $m=1$ and the data (ii), a phase with $L'\simeq\kappa L^2$ and $\kappa=1.01$, $1.00$ at $L=10$, $10^2$ appears, in agreement with $\kappa_1=1$ (Proposition~\ref{prop:theta-selection}). The phase is reproduced for all grids and windows tested ($S=30,45,60,90$). Beyond $\tau\approx7$ these runs no longer agreed: some left the class, others continued to grow. We repeated these runs in extended precision (long double, $18$ digits). The phase $\kappa\simeq1$ then extended to $L\approx3\cdot10^2$, which pointed to a loss of precision. We traced the loss to two sources:
\begin{itemize}
\item \emph{The initial data.} Computed by finite differences, they lose all relative precision deep in the ridge layers, where $\rho-\rho_0=O(e^{2s})$ and $\Omega=O(e^{s})$. We now evaluate them analytically, writing $\cos2k\theta$ and $\sin2k\theta$ as Chebyshev polynomials of $\cos2\theta=-\tanh s$ times powers of $\sin2\theta=\mathrm{sech}\,s$. Then $\rho$, $\sigma$ and $\Omega$ have full relative precision for all $s$.
\item \emph{The discrete solvability condition.} The discrete value of $C:=\int\nu\,\mathrm{sech}(s)\,\Omega\,ds$ carries a residual of order $10^{-17}$ from the quadrature. The scheme conserves it, and it competes with the jet term $6\mathcal W_2C_1=O(e^{-4\tau})$ once $\tau\approx10$. We remove it by choosing the level so that $\frac{dC}{dt}=-\lambda_cLC$, i.e.\ $L=-F|_{L=0}/(\partial F/\partial L+\lambda_cC)$, with $\lambda_c=10$.
\end{itemize}
With both changes, in double precision, the Boussinesq case behaves like $m=2$. For the data (ii), $\kappa=1.0004$, $1.0000$, $1.0000$ at $L=10^3,10^4,10^5$ ($\tau=8.5$, $10.8$, $13.1$; blow-up time $T=2.92421$). These values agree for the windows $S=45$ and $60$ and for $\lambda_c=10$ and $20$. For the data (iii), $\kappa=1.0022$, $1.0002$, $1.0000$ at $L=10^3,10^4,10^5$. For $\mathcal W=\cos2\theta$, $\rho=1+\frac12\cos2\theta$, $\kappa=1.0022$, $0.9998$ at $L=10^3,10^4$. For $m=2$ the improved solver reproduces Table~\ref{tab:theta-kappa}. For $m=1.1$ it gives $\kappa=1.4476$ at $L=10^4$ and $10^5$, against $\kappa_{1.1}=1.448$. For the data (i) and (iv) the solution still leaves the class early ($\tau\lesssim2.6$), as for the data of (c) with $m=2$. The computations thus support a blow-up of strain type with the selected rate $\kappa_1=1$ for the Boussinesq case too: $v_a\simeq(T-t)^{-1}$ and, since $a/\kappa_1=1$, $u_a^2=\rho_0\propto L$, i.e.\ $u_a\propto(T-t)^{-1/2}$.
\end{enumerate}
\end{remark}

\begin{table}[ht]
\centering\small
\begin{tabular}{lcccccc}
\hline
data $(\mathcal W;\rho)$, $m=2$ & $T$ & $L=10$ & $10^2$ & $10^3$ & $10^4$ & $10^5$\\
\hline
(i) $-\cos2\theta$; $1+0.9\cos2\theta$ & $0.7609$ & $1.79$ & $3.17$ & $4.08$ & $4.51$ & --\\
(ii) $-\cos2\theta$; $1-0.9\cos2\theta$ & $1.0095$ & $2.54$ & $3.97$ & $4.50$ & $4.74$ & $4.87$\\
(iii) $-\cos2\theta+0.3\cos4\theta$; $1+0.5\cos4\theta$ & $1.5402$ & $3.05$ & $4.19$ & $4.57$ & $4.76$ & $4.86$\\
(iv) $-2\cos2\theta$; $0.2+0.1\cos2\theta$ & $1.9721$ & $2.20$ & $3.77$ & $4.39$ & $4.68$ & --\\
\hline
\end{tabular}
\caption{The coefficient $\kappa=-\frac{d}{dt}(1/L)$ at the times when the level reaches the given values, for the axisymmetric Euler case $m=2$. The initial profiles of $\mathcal W$ are given up to their level, which is fixed by $F=0$. $T$ is the blow-up time.}\label{tab:theta-kappa}
\end{table}

The value of $\kappa$ can be predicted by a formal asymptotic argument, which we now give. It shows that the late dynamics is controlled by the Taylor coefficients of the profiles at the plane ridge, as the sensitivity to the plane-ridge layer suggests.

\begin{proposition}[Selection of the blow-up rate; formal]\label{prop:theta-selection}
Let $a:=\frac{2m^2}{m+1}$ and $\beta:=\frac{2(m-1)}{m+1}$. Consider a solution of \eqref{eq:theta-logangle} with $L(t)\to+\infty$, $\mathcal W=L(1+o(1))$ uniformly and $\rho_0/L^2\to0$. Let $\tau:=\int^tL\,dt'$ and $\kappa:=L'/L^2=d\log L/d\tau$. Then, at the level of formal asymptotics,
\[
\kappa\to\kappa_m:=a-4+\gamma_m,\qquad \gamma_m:=\beta+\min(2m+2,\,6),
\]
that is,
\[
\kappa_1=1,\qquad \kappa_m=\frac{4m^2+2m-4}{m+1}\quad(1\le m\le2),\qquad \kappa_m=\frac{2m(m+2)}{m+1}\quad(m\ge2),
\]
with $\kappa=\frac{16}3-\frac1\tau+o(\frac1\tau)$ in the marginal case $m=2$. Consequently, in the isotropic limit $v_a=V=\frac{2L}{m+1}$,
\[
\frac{v_a'}{v_a^2}\to c_m:=\frac{(m+1)\kappa_m}{2}=\begin{cases}1,&m=1,\\ m(m+2),&m\ge2,\end{cases}
\]
and $\rho_0\propto L^{a/\kappa_m}=L^{m/(m+2)}$ for $m\ge2$.
\end{proposition}

\begin{proof}[Derivation]
In the variable $\tau$, the equations \eqref{eq:theta-logangle} reduce at leading order to translation with speed $2$ in $s$, with growth rate $a$ for $\rho$ and decay rate $\beta$ for $\Omega$. A profile that is even in $\theta$ at the plane ridge is a series in $e^{2s}$ as $s\to-\infty$; one that is odd is a series in $e^{s},e^{3s},\dots$. After translation by $2\tau$, the region $s=O(1)$ therefore carries the Taylor expansion at the plane ridge, and the original bulk has moved to $s\approx2\tau$. Near $\theta=0$ one has $\Omega=\mathcal L\mathcal W=6\mathcal W_2\theta+O(\theta^3)$. Hence, for $s\ll2\tau$,
\[
\Omega\simeq 6\mathcal W_2(\tau)\,e^{s}+A_3(\tau)\,e^{3s}+\dots,\qquad A_3=O\bigl(e^{-(\beta+6)\tau}\bigr),
\]
where $\mathcal W_2$ evolves by the apex equations at the ridge ends of Proposition~\ref{prop:theta-ridge}. The solvability condition $\int\nu\,\mathrm{sech}(s)\,\Omega\,ds=0$ is exact, and $\nu\,\mathrm{sech}(s)\sim2e^{-(m+1)s}$ as $s\to+\infty$. The term $6\mathcal W_2e^s$ contributes $6\mathcal W_2C_1$ with $C_1=\int\nu\,\mathrm{sech}(s)e^s\,ds>0$. The term $A_3e^{3s}$ contributes $O(e^{-(\beta+6)\tau})$ if $m\ge3$ and $O(\tau e^{-(\beta+6)\tau})$ if $m=2$. If $m=1$ it is absorbed in the bulk, which contributes $O(e^{-(2m+2+\beta)\tau})$. The solvability condition thus forces $\mathcal W_2=O(e^{-\gamma_m\tau})$, with an extra factor $\tau$ if $m=2$. In the variable $\tau$, the apex equations at the ridge ends read
\[
\frac{d\mathcal W_2}{d\tau}=-(\beta+2)\,\mathcal W_2-\frac{m+1}6\,\frac{\rho_2}{L},\qquad \frac{d\rho_2}{d\tau}=(a-4)\,\rho_2+o(\rho_2),
\]
since $\frac{4m}{m+1}=\beta+2$ and $\rho_0\mathcal W_2/L=o(\rho_2)$. Because $\gamma_m>\beta+2$, a decay of $\mathcal W_2$ at the rate $\gamma_m$ is possible only if the forcing balances it: $\frac{\rho_2}{L}\simeq\frac{6}{m+1}(\gamma_m-\beta-2)\mathcal W_2$. Then $L\propto\rho_2/\mathcal W_2\propto e^{(a-4+\gamma_m)\tau}$, i.e.\ $\kappa\to a-4+\gamma_m$. The factor $\tau$ for $m=2$ gives the correction $-1/\tau$. Finally $\rho_0'=aL\rho_0$ gives $\rho_0\propto e^{a\tau}\propto L^{a/\kappa_m}$.
\end{proof}

The prediction agrees with the computations. For $m=2$ the values of Table~\ref{tab:theta-kappa} approach $\frac{16}3$ slowly, as the logarithmic correction predicts: $(\frac{16}3-\kappa)^{-1}$ grows with $\tau$ at a slope between $0.8$ and $1$. The ridge exponent $d\log\rho_0/d\log L$ decreases towards $\frac12$ ($0.55$ at $L=10^5$). For the formal values $m=3,4$ of the family (see Conjecture~\ref{conj:4D} below for $m=3$): for $m=3$ ($\kappa_3=7.5$) we find $\kappa=7.23$, $7.42$ at $L=10^5$ for the data (ii), (iii), with differences to $7.5$ that decrease exponentially in $\tau$. For $m=4$ ($\kappa_4=9.6$) we find $9.18$, $9.46$. For $m=1$ the value $\kappa\simeq1.01$ was observed for $10\le L\le80$ before the computation lost accuracy (Remark~\ref{rem:theta-numerics}(d)).
Figure~\ref{fig:theta-jet} checks the mechanism directly.
\begin{itemize}
\item \emph{The jet region (panel (a)).} For $m=3$ at $L=10^3$, the quantities $\Omega e^{-s}$ and $2(\rho-\rho_0)e^{-2s}$ are constant, to $2\cdot10^{-3}$, on $-6<s<-2$. This region carries the Taylor expansion at the plane ridge. The value of $\mathcal W_2$ obtained from $A_1=6\mathcal W_2$ agrees within $2.5\%$ with the value obtained from $\mathcal W$ directly. The original bulk sits near $s=2\tau$.
\item \emph{The balance (panel (b)).} The derivation predicts $\rho_2/(L\mathcal W_2)\to\frac{6}{m+1}(\gamma_m-\beta-2)=\frac{24}{m+1}$ for $m\ge2$. The computed ratios approach these values, slowly for $m=2$.
\item \emph{The decay of $\mathcal W_2$.} Over the last decade of $L$, $d\log\mathcal W_2/d\log L$ equals $-0.95$, $-0.77$ and $-1.31$ for $m=3,4,2$. The predicted values $-\gamma_m/\kappa_m$ are $-\frac{14}{15}$, $-\frac34$, and $-\frac54$ up to the logarithmic correction.
\end{itemize}

\begin{figure}[ht]
\centering
\includegraphics[width=\textwidth]{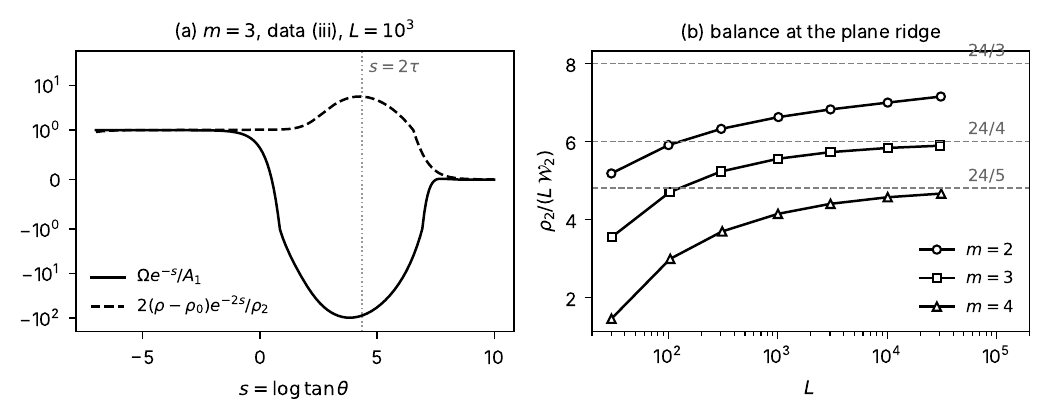}
\caption{The plane-ridge mechanism of Proposition~\ref{prop:theta-selection}. (a) $m=3$, data (iii), at $L=10^3$ ($\tau=2.18$). Left of $s\approx0$, the normalized profiles are constant: they carry the Taylor expansion at the plane ridge, with $A_1=6\mathcal W_2$. The original bulk has been translated to $s\approx2\tau$. (b) The ratio $\rho_2/(L\mathcal W_2)$ against $L$, for $m=2$ (data (ii)) and $m=3,4$ (data (iii)), and the predicted limits $\frac{24}{m+1}$ (dashed).}\label{fig:theta-jet}
\end{figure}

\begin{remark}[Jet closure]\label{rem:theta-jet-closure}
In the regime of Proposition~\ref{prop:theta-selection}, take $m\ge3$; for $m=2$ the same holds up to logarithmic factors. The late dynamics then reduces to the apex equations at the plane end of the ridge, closed by the balance derived above:
\[
\rho_0'=aL\rho_0,\qquad \rho_2'=(a-4)L\rho_2,\qquad \mathcal W_2'=-(\beta+6)L\,\mathcal W_2,\qquad L=\frac{(m+1)\,\rho_2}{24\,\mathcal W_2}.
\]
Differentiating the last relation gives $L'=(a+\beta+2)L^2=\kappa_mL^2$. In this regime the level is therefore determined by the two Taylor coefficients $\rho_2$ and $\mathcal W_2$ at the plane ridge. This replaces the free anisotropy of the $\theta$-independent apex dynamics (Proposition~\ref{prop:theta-structure}(i)). Since $L>0$, the regime requires $\rho_2\mathcal W_2>0$. The computations reach it from data with either sign of $\rho_2\mathcal W_2$ at $t=0$, e.g.\ data (i), so this sign condition does not by itself predict which data blow up.
\end{remark}

The derivation of Proposition~\ref{prop:theta-selection} has two parts. The first is a decay estimate for $\mathcal W_2$ forced by the solvability condition. The second converts this decay into the rate of $L$ through the exact apex equations at the ridge ends. The second part can be made rigorous.

\begin{proposition}[Reduction of the rate selection to a decay estimate]\label{prop:theta-reduction}
Let $\lambda=a=b=1$, $m\ge1$, and let $(\rho,\mathcal W)$ be a solution of \eqref{eq:theta-transport} on $[0,T)$, smooth at the ridges, with $L(t)>0$ and $\tau(t)=\int_0^tL\,dt'\to\infty$ as $t\to T$. Assume:
\begin{enumerate}[label=(H\arabic*)]
\item $\mathcal W_0/L\to1$ and $\rho_0\mathcal W_2/(L\rho_2)\to0$ as $\tau\to\infty$;
\item there are $\gamma>\beta+2$ and $c\ne0$ such that $e^{\gamma\tau}\mathcal W_2\to c$ and $e^{\gamma\tau}\frac{d\mathcal W_2}{d\tau}\to-\gamma c$.
\end{enumerate}
Then
\[
\frac{\rho_2}{L}=\frac{6}{m+1}(\gamma-\beta-2)\,c\,e^{-\gamma\tau}\bigl(1+o(1)\bigr),\qquad \lim_{\tau\to\infty}\frac{\log L}{\tau}=a-4+\gamma,
\]
and $c\rho_2>0$ for large $\tau$. If in addition $e^{-(a-4)\tau}\rho_2$ and $e^{\gamma\tau}\bigl(\frac{d\mathcal W_2}{d\tau}+(\beta+2)\frac{\mathcal W_0}L\mathcal W_2\bigr)$ converge to nonzero limits, then $\log L-(a-4+\gamma)\tau$ converges. Hence $L\simeq\bigl(\kappa(T-t)\bigr)^{-1}$ with $\kappa=a-4+\gamma$.
\end{proposition}

\begin{proof}
Divide the equations of Proposition~\ref{prop:theta-ridge} by $L$ and use $\frac{4m}{m+1}=\beta+2$, $\frac{6m^2}{m+1}=3a$:
\[
\frac{d\mathcal W_2}{d\tau}=-(\beta+2)\frac{\mathcal W_0}L\,\mathcal W_2-\frac{m+1}6\,\frac{\rho_2}L,\qquad \frac{d\rho_2}{d\tau}=(a-4)\frac{\mathcal W_0}L\,\rho_2+3a\,\frac{\rho_0\mathcal W_2}{L}.
\]
By the first equation and (H1)--(H2), $\frac{\rho_2}L=-\frac6{m+1}\bigl(\frac{d\mathcal W_2}{d\tau}+(\beta+2)(1+o(1))\mathcal W_2\bigr)=\frac6{m+1}(\gamma-\beta-2+o(1))\,c\,e^{-\gamma\tau}$. Since $L>0$ and $\gamma>\beta+2$, this gives $c\rho_2>0$. By the second equation and (H1), $\frac{d}{d\tau}\log|\rho_2|=a-4+o(1)$, so $\log|\rho_2|=(a-4)\tau+o(\tau)$. Therefore $\log L=\log|\rho_2|+\gamma\tau+O(1)=(a-4+\gamma)\tau+o(\tau)$. Under the additional assumptions both $\log|\rho_2|-(a-4)\tau$ and $\log|\rho_2/L|+\gamma\tau$ converge, hence so does $\log L-(a-4+\gamma)\tau$. Then $L=Ce^{\kappa\tau}(1+o(1))$, and $T-t=\int_\tau^\infty L^{-1}d\tau'=(\kappa L)^{-1}(1+o(1))$.
\end{proof}

With $\gamma=\gamma_m=\beta+\min(2m+2,6)$, Proposition~\ref{prop:theta-reduction} gives $\kappa_m$ of Proposition~\ref{prop:theta-selection}. Figure~\ref{fig:theta-jet}(b) and the decay rates of $\mathcal W_2$ reported there are numerical checks of (H2). What remains to be proved is (H1)--(H2) from the solvability condition. In the variable $s$ this condition reads $\int\nu\,\mathrm{sech}(s)\,\Omega\,ds=0$. Write it as $6\mathcal W_2C_1+R(\tau)=0$, where $C_1=\int\nu\,\mathrm{sech}(s)\,e^s\,ds$ and $R$ collects the remainder of the Taylor expansion at the plane ridge and the bulk near $s=2\tau$. Three steps remain:
\begin{enumerate}[label=(\alph*)]
\item a representation of $\Omega$ along the characteristics $\frac{ds}{d\tau}=\frac{2\mathcal W}{L}$, in which the bulk is translated and damped at the rate $\beta$;
\item the estimate $R(\tau)=r\,e^{-\gamma_m\tau}(1+o(1))$ with $r\ne0$ (and an extra factor $\tau$ for $m=2$);
\item a bootstrap that controls the deviation $\mathcal W/L-1$, which enters (a) and (H1).
\end{enumerate}
Step (b) for $m<2$ uses the bulk, and for $m>2$ the next Taylor coefficient at the plane ridge. This is the dichotomy of Remark~\ref{rem:role-of-m}. The sign condition $c\rho_2>0$ is the condition $\rho_2\mathcal W_2>0$ of Remark~\ref{rem:theta-jet-closure}.

\begin{remark}[Saddle structure of the pressure]\label{rem:theta-saddle}
The author pointed out, with a paper model of the level surfaces of $P$, that $P_{xx}$ and $P_{\theta\theta}$ have opposite signs at the ridges. The blow-up of Proposition~\ref{prop:theta-selection} has exactly this structure. Let $v_a'=c\,v_a^2$ in the isotropic limit, with $u_a^2/v_a^2\to0$. Then $V=v_a$ and $G=mv_a$, and the pressure terms are $A_2=V_t-V^2=(c-1)v_a^2$ and $A_3=G_t+G^2=m(c+m)v_a^2$. With $P=x^2\Pi(\theta)$ and $2\Pi=\cos^2\theta A_2-\sin^2\theta A_3$,
\[
\begin{aligned}
\Pi(\theta)&=\tfrac12v_a^2\bigl((c-1)\cos^2\theta-m(c+m)\sin^2\theta\bigr),\\
\Pi''(\theta)&=-(m+1)(c+m-1)\,v_a^2\cos(2\theta).
\end{aligned}
\]
Hence, since $c\ge1$:
\begin{itemize}
\item on the plane, $P_{xx}=(c-1)v_a^2\ge0$ and $P_{\theta\theta}=-(m+1)(c+m-1)x^2v_a^2<0$;
\item on the axis, $P_{xx}=-m(c+m)v_a^2<0$ and $P_{\theta\theta}=(m+1)(c+m-1)x^2v_a^2>0$.
\end{itemize}
The level surfaces of $P$ are saddles along both ridges. The zero level is the cone $\tan^2\theta=\frac{c-1}{m(c+m)}$, i.e.\ $\theta\approx31^\circ$ for $m=2$, $c=8$. In the notation of the apex dynamics, $\pi_r=(c-1)v_a^2$ and $\pi_z=-m(c+m)v_a^2$, consistent with the trace identity $m\pi_r+\pi_z=-m(m+1)v_a^2$ at $u_a=0$. The anisotropy is
\[
\mu=-\frac{\pi_z}{\pi_r}=\frac{m(c+m)}{c-1}=m+\frac{m(m+1)}{c-1},
\]
so $\mu=\frac{m^2(m+3)}{m^2+2m-1}$ for $m\ge2$, e.g.\ $\mu=\frac{20}7$ for $m=2$.
It lies above the window $(0,m)$ of \cite[Theorem~13.6]{ShiI}, consistent with a blow-up in which the swirl is subdominant. For $m=1$, $c=1$: then $\pi_r=0$ and $\Pi=-v_a^2\sin^2\theta$, so the plane is a degenerate saddle, and the apex dynamics is $v_a'=v_a^2$.
\end{remark}

%=====================================================================
\section{The parameter \texorpdfstring{$m$}{m}, the swirl and related work}\label{subsec:theta-dimension}
%=====================================================================

The value $m=1$ is the 2D inviscid Boussinesq system and $m=2$ is the 3D axisymmetric Euler system. The values $m=3,4$ used above are, so far, formal members of the family $(Em)$. The author conjectures that $m=3$ corresponds to the 4D axisymmetric Euler equations. For the swirl-free part of the system this holds in every dimension.

\begin{proposition}[Swirl-free sector]\label{prop:swirlfree-dim}
Let $d\ge2$ and consider velocity fields on $\R^d$ that are axisymmetric about the $z$-axis and swirl-free, i.e.\ $\mathbf u=u^r(t,r,z)\,e_r+u^z(t,r,z)\,e_z$ with $r$ the distance to the axis; for $d=2$, take $r=|x_2|$ with the symmetries and variables of \cite[(2.7)]{ShiI}. Then $v:=-u^r/r$, $g:=u^z/z$ and the pressure $p$ satisfy $(Em)$ with $u\equiv0$ and $m=d-1$, together with the pressure equation \cite[(2.10)]{ShiI}. Conversely, a solution of $(Em)$ with $u\equiv0$ and $m=d-1$ defines such a solution of the $d$-dimensional Euler equations.
\end{proposition}

\begin{proof}
With $\frac{\D}{\D t}=\partial_t+u^r\partial_r+u^z\partial_z=\partial_t-vr\partial_r+gz\partial_z$ and $\frac{\D r}{\D t}=-rv$, $\frac{\D z}{\D t}=zg$, the momentum equations $\frac{\D u^r}{\D t}=-p_r$ and $\frac{\D u^z}{\D t}=-p_z$ become $\frac{\D v}{\D t}=v^2+\frac1rp_r$ and $\frac{\D g}{\D t}=-g^2-\frac1zp_z$. These do not depend on $d$. The divergence in $\R^d$ is $\partial_ru^r+\frac{d-2}ru^r+\partial_zu^z=z g_z-rv_r+g-(d-1)v$, which is \eqref{eq:Em}(4) with $m=d-1$. The Laplacian of an axisymmetric function is $\partial_r^2+\frac{d-2}r\partial_r+\partial_z^2$, which is the operator in \cite[(2.10)]{ShiI}. In the source $\mathrm{tr}(\nabla\mathbf u)^2$, the eigenvalue $u^r/r=-v$ has multiplicity $d-2$, which gives the term $(m-1)v^2$.
\end{proof}

The $\theta$-dependent dynamics is consistent with this. The weight $\cos^{m-1}\theta$ is $r^{d-2}$, the angular part of the volume element. The harmonic profile $x^2\Pi_0=(d-1)z^2-r^2$ is harmonic in $\R^d$. The stretching $-(m-1)V\Omega=(d-2)\frac{u^r}{r}\Omega$ in \eqref{eq:theta-transport} is the stretching term $(d-2)\frac{u^r}{r}\omega$ of the vorticity equation of axisymmetric, swirl-free flows in $\R^d$ \cite{MillerTsai2026}.

\begin{conjecture}[Author]\label{conj:4D}
The system $(E3)$ corresponds to the 4D axisymmetric Euler equations, in the same way that $(E1)$ and $(E2)$ correspond to the 2D Boussinesq and the 3D axisymmetric Euler equations.
\end{conjecture}

\begin{remark}[The swirl sector]\label{rem:swirl-sector}
By Proposition~\ref{prop:swirlfree-dim}, Conjecture~\ref{conj:4D} holds for $u\equiv0$. The open part is the swirl. The equation $\frac{\D u}{\D t}=\frac{m^2}2uv$ is equivalent to the conservation of $r^{m^2/2}u$ along particle paths, since $\frac{\D r}{\D t}=-rv$. For $m=1$ this is the transport of $\vartheta=x_2u^2$. For $m=2$ it is the conservation of the circulation $r^2u=rv^\phi$. For $m=3$ it would be the conservation of $r^{9/2}u$. A flow on $\R^4$ that is invariant under the rotations $SO(3)$ about an axis has no swirl: the orbits are 2-spheres, and a tangent vector field on $S^2$ that is invariant under $SO(3)$ vanishes. A proof of the conjecture therefore needs either a larger class of 4D flows or an interpretation of $u$ as an active scalar, as for $m=1$. The swirl cannot be dropped from the blow-up of Proposition~\ref{prop:theta-selection}. In computations with $\rho\equiv0$ and the data $\mathcal W(0)=\mp\cos2\theta$, $-\cos2\theta+0.3\cos4\theta$, $\frac12\cos4\theta$, the solutions of \eqref{eq:theta-logangle} for $m=1,2,3$ show no growth up to $t=30$: $|\mathcal W_0|$ and $|\hat{\mathcal W}_0|$ decay roughly like $t^{-1}$. The jet closure of Remark~\ref{rem:theta-jet-closure} explains this: $L\propto\rho_2/\mathcal W_2$ vanishes with $\rho$. The swirl is subdominant at the blow-up, $\rho_0/L^2\to0$, but it is what selects and drives the level. For $d\ge4$ the swirl-free equations have been studied by Miller and Tsai \cite{MillerTsai2026}, who prove a conditional blow-up result, and by Miller \cite{Miller2026}.
\end{remark}

\begin{remark}[Swirl-free regularity]\label{rem:swirlfree-regular}
The need for swirl matches the classical theory. For axisymmetric, swirl-free 3D flows, Ukhovskii and Yudovich \cite{UY1968} (Euler) and, independently, Ladyzhenskaya \cite{Lady1968} (Navier--Stokes) proved global existence and uniqueness. Their condition on the data is that $\omega^\phi/r$ is bounded and square integrable, i.e.\ the vorticity vanishes fast enough at the axis. For the 3D Euler equations, finite-time blow-up therefore requires swirl, or data outside this class.

The homogeneous class \eqref{eq:homog} lies outside it. There the vorticity is homogeneous of degree zero, so $\omega^\phi/r$ is of order $1/r$ near the axis. Elgindi \cite{Elgindi21} showed that swirl-free axisymmetric solutions with $C^{1,\alpha}$ velocity can blow up, so regularity of the swirl-free homogeneous solutions is not a consequence of \cite{UY1968}. In our computations, however, the swirl-free homogeneous solutions of \eqref{eq:theta-logangle} with $m=2$ do not grow (Remark~\ref{rem:swirl-sector}), and the jet closure gives the reason: the level is selected by $\rho_2$. Elgindi and Jeong made the same observation for their scale-invariant solutions in domains with boundary \cite{EJ2019}: solutions with swirl become singular, and solutions with zero swirl are globally regular. In the blow-up of Proposition~\ref{prop:theta-selection} the swirl is subdominant at the singular time, $u_a^2/v_a^2\to0$, but it is indispensable: it is the swirl at the plane ridge that selects and drives the strain. The author emphasized this point.
\end{remark}

\begin{remark}[Comparison with the domains of Elgindi and Jeong]\label{rem:EJ-domains}
The scale-invariant blow-up results of Elgindi and Jeong are posed in domains with boundary. For the Boussinesq system \cite{EJ2020} the domain is the sector $\{x_1\ge\gamma|x_2|\}$, $\gamma>0$, of opening angle $2\arctan(1/\gamma)<\pi$. The half-plane, i.e.\ the angular interval $0<\theta<\pi$, is the limit $\gamma\to0$, which they do not cover. For the axisymmetric Euler equations with swirl \cite{EJ2019} the domain is $\{(1+\varepsilon|z|)^2\le x^2+y^2\}$. It excludes the axis, and its boundary has the shape of an hour-glass. The scale-invariant solutions are centred at the corner circle $r=1$, $z=0$, where the meridian domain is a sector of opening angle $2\arctan(1/\varepsilon)=\pi-2\arctan\varepsilon<\pi$. In both papers the scale-invariant solutions satisfy a $(1+1)$-dimensional system in time and angle, whose solutions blow up in finite time. These systems are the counterparts of the $\theta$-dependent apex dynamics of this paper. Their solutions have Lipschitz velocity and are smooth away from the corner. This is the regularity of the homogeneous class \eqref{eq:homog}, so the comparison concerns the geometry, not the regularity. The $C^{1,\alpha}$ blow-up results are those of Elgindi \cite{Elgindi21} (3D Euler on $\R^3$, without swirl) and of Chen and Hou \cite{CH2021} (with boundary). In both cases the angular problem of \cite{EJ2019,EJ2020} lives on an interval bounded by walls with a no-flux condition. In our setting the angle ranges over the full circle, $\theta\in[0,2\pi)$, without boundary. The symmetries of \cite[Section~2]{ShiI} reduce it to $[0,\frac\pi2]$, whose ends $\theta=0,\frac\pi2$ are symmetry axes (the flat ridges), not walls. On this interval the harmonic profile $\Pi_0$ is admissible. This produces the free level $L$, the solvability condition $\int\mu\Omega=0$ and the constraint $F=0$, and through them the selection mechanism of Proposition~\ref{prop:theta-selection}. Whether a free level of this kind can arise in the sector problems depends on the opening angle and the symmetry class; we do not pursue this comparison here. The author suggested this comparison.
\end{remark}

\begin{remark}[The role of $m$: inflow through the plane, outflow along the axis]\label{rem:role-of-m}
The author explains the difference between $m=1$ and $m=2$ by the geometry of the flow near the stagnation point. In 3D the fluid comes in radially through the plane $z=0$, which is two-dimensional, and has to leave along the axis, which is one-dimensional. In 2D it comes in along one line and leaves along another. The $\theta$-dependent dynamics makes this picture quantitative.
\begin{enumerate}[label=(\alph*)]
\item \emph{Outflow.} At the apex the divergence condition gives $g=mv$. The axial outflow strain is $m$ times the radial inflow strain: the planar strain is balanced for $m=1$, and the outflow is twice the inflow for $m=2$.
\item \emph{Angular flow.} For $\mathcal W>0$ the angular flow $\dot\theta=\sin(2\theta)\mathcal W$ carries material from the plane to the axis. In the variable $s$ this is a translation with speed $2\mathcal W$.
\item \emph{Spin-up and stretching.} Along this flow the swirl is amplified at the rate $a=\frac{2m^2}{m+1}$ in $\tau$ ($a=1$ for $m=1$, $a=\frac83$ for $m=2$). The vorticity profile is stretched at the rate $\beta=\frac{2(m-1)}{m+1}$ ($\beta=0$ for $m=1$, $\beta=\frac23$ for $m=2$). Both come from the conservation laws $r^{m^2/2}u=\mathrm{const}$ and $(d-2)\frac{u^r}{r}\omega$, $d=m+1$, along the inflow.
\item \emph{Thinness of the axis.} In the solvability condition $\int\nu\,\mathrm{sech}(s)\,\Omega\,ds=0$ the weight decays like $e^{-(m+1)s}$ towards the axis, reflecting the volume element $r^{m-1}$. The material carried to the axis therefore loses influence faster for larger $m$. This is the competition in $\gamma_m=\beta+\min(2m+2,6)$. For $m<2$ the decay of $\mathcal W_2$, and hence the level, is fixed by the material that has been carried to the axis, a nonlocal and data-dependent quantity. For $m>2$ it is fixed by the next Taylor coefficient at the plane ridge, a local quantity. The 3D Euler equations, $m=2$, are exactly the threshold, where both contribute and the correction is logarithmic.
\end{enumerate}
Since $m$ enters \eqref{eq:theta-logangle} only through coefficients, the equations make sense for real $m\ge1$, and we used this to test the formula of Proposition~\ref{prop:theta-selection} between the physical values. Table~\ref{tab:theta-kappa-m} compares $\kappa_m=\frac{4m^2+2m-4}{m+1}$ ($1\le m\le2$) and $\kappa_m=\frac{2m(m+2)}{m+1}$ ($m\ge2$) with the computed values. The two formulas agree at $m=2$. The agreement is good away from $m=2$, and the approach to $\kappa_m$ is slow near $m=2$ from both sides, as expected at the threshold. Thus the computations support the following picture. The rate of the strain blow-up increases with $m$. The mechanism is local at the plane ridge for $m\ge2$. In the Boussinesq case $m=1$, with no vortex stretching ($\beta=0$) and the weakest spin-up ($a=1$), the selected rate is $\kappa_1=1$ and the selection is nonlocal.
\end{remark}

\begin{table}[ht]
\centering\small
\begin{tabular}{lccccccccc}
\hline
$m$ & $1$ & $1.1$ & $1.25$ & $1.5$ & $1.8$ & $2$ & $2.5$ & $3$ & $4$\\
\hline
$\kappa_m$ (formal) & $1$ & $1.448$ & $2.111$ & $3.200$ & $4.486$ & $5.333$ & $6.429$ & $7.500$ & $9.600$\\
computed & $1.000$ & $1.448$ & $2.10$ & $3.14$ & $4.17$ & $4.87$ & $6.14$ & $7.42$ & $9.46$\\
at $L\approx$ & $10^5$ & $10^5$ & $10^4$ & $10^4$ & $10^4$ & $10^5$ & $10^4$ & $10^5$ & $10^5$\\
\hline
\end{tabular}
\caption{The selected rate $\kappa_m$ of Proposition~\ref{prop:theta-selection} against the computed $\kappa=-\frac{d}{dt}(1/L)$, for data (ii) if $m\le1.1$ and otherwise the larger of the values for data (ii) and (iii). The values for $m=1,1.1$ are from the improved computation of Remark~\ref{rem:theta-numerics}(d). For $m\in(1.8,2.5)$ the approach to $\kappa_m$ is slow (threshold at $m=2$).}\label{tab:theta-kappa-m}
\end{table}

These computations and Proposition~\ref{prop:theta-selection} indicate that, within the homogeneous class \eqref{eq:homog}, the closed $\theta$-dependent apex dynamics of the 3D axisymmetric Euler equations ($m=2$) can blow up in finite time, and by Remark~\ref{rem:theta-numerics}(d) so can that of the 2D Boussinesq equations ($m=1$). The blow-up is of strain type: $v_a\simeq\bigl(8(T-t)\bigr)^{-1}$, $u_a\propto(T-t)^{-1/4}$ for $m=2$, and $v_a\simeq(T-t)^{-1}$, $u_a\propto(T-t)^{-1/2}$ for $m=1$. It is driven by the swirl (the buoyancy for $m=1$), although the swirl is subdominant (Remark~\ref{rem:swirl-sector}). It is carried by the level $L$, i.e.\ by the harmonic part of the pressure that the constraint $F=0$ selects, and the profile becomes isotropic. Its rate is selected by the Taylor coefficients at the plane ridge through the solvability condition. Two caveats apply. First, the class \eqref{eq:homog} has infinite energy and is not smooth at the origin, and $F=0$ is the condition for remaining in this class, not a property of general solutions. Second, Proposition~\ref{prop:theta-selection} is formal. Its second half is proved in Proposition~\ref{prop:theta-reduction}; what remains is the decay estimate (H2), which requires control of the plane-ridge layer and of the bulk at $s\approx2\tau$.

%=====================================================================
\section{General parameters \texorpdfstring{$(\lambda,a,b)$}{(lambda,a,b)} and the anisotropy at the ridge ends}\label{subsec:theta-general}
%=====================================================================
We now drop the assumption $\lambda=a=b=1$ of Section~\ref{subsec:theta-derivation}. The divergence relation \eqref{eq:VG-from-W} does not involve $(\lambda,a,b)$, so $V$ and $G$ are still given by $\mathcal W=\frac12(G+V)$ through \eqref{eq:VG-from-W}. Put
\[
A_2:=\lambda V_t-V^2+\rho+\sin(2\theta)\,\mathcal WV_\theta,\qquad A_3:=\lambda G_t+G^2+\sin(2\theta)\,\mathcal WG_\theta ,
\]
and
\[
\mathcal L_{a,b}f:=\bigl(b\cos^2\theta+a\sin^2\theta\bigr)\bigl(\sin(2\theta)f'\bigr)'+2\bigl(b\cos^2\theta-am\sin^2\theta\bigr)f' .
\]

\begin{proposition}[General parameters]\label{prop:theta-general}
Let $\lambda,a,b>0$ and let \eqref{eq:homog} hold.
\begin{enumerate}[label=(\roman*)]
\item $\lambda\rho_t+\sin(2\theta)\,\mathcal W\rho_\theta=m^2V\rho$.
\item The pressure profile satisfies $2\Pi=\cos^2\theta\,\frac{A_2}a-\sin^2\theta\,\frac{A_3}b$ and $\Pi_\theta=-\sin\theta\cos\theta\bigl(\frac{A_2}a+\frac{A_3}b\bigr)$. The condition $P_{x\theta}=P_{\theta x}$ is $b\cos^2\theta\,\partial_\theta A_2=a\sin^2\theta\,\partial_\theta A_3$.
\item This condition is linear in $\mathcal W_t$:
\[
\lambda\,\mathcal L_{a,b}[\mathcal W_t]=-(m+1)\Bigl(b\cos^2\theta\,\partial_\theta\bigl(A_2-\lambda V_t\bigr)-a\sin^2\theta\,\partial_\theta\bigl(A_3-\lambda G_t\bigr)\Bigr).
\]
For $a=b$, $\mathcal L_{a,a}=a\mathcal L$.
\item Let
\[
w_{a,b}:=\frac{\sin\theta\cos^m\theta}{(b\cos^2\theta+a\sin^2\theta)^{(m+3)/2}},\qquad q_{a,b}:=\frac{2\sin^2\theta\cos^{m+1}\theta}{(b\cos^2\theta+a\sin^2\theta)^{(m+1)/2}}.
\]
Then $w_{a,b}\,\mathcal L_{a,b}f=(q_{a,b}f')'$. Hence the regular functions in the kernel of $\mathcal L_{a,b}$ are the constants. The average of (iii) with the weight $w_{a,b}$ contains no time derivative: it is a constraint, the analogue of $F=0$, and it fixes the level of $\mathcal W$.
\item The parameter $\lambda$ is a time scale: $(\rho,\mathcal W)(t,\theta)$ solves the system with parameter $\lambda$ if and only if $(\rho,\mathcal W)(\lambda s,\theta)$ solves it with $\lambda=1$ in the time $s$. For $a=b$ the condition in (ii) is the one for $a=b=1$, and $\Pi$ is multiplied by $1/a$. Therefore the results of Sections~\ref{subsec:theta-derivation}--\ref{subsec:theta-dimension} hold for every $\lambda>0$ and $a=b>0$, with $t$ replaced by $t/\lambda$.
\end{enumerate}
\end{proposition}

\begin{proof}
(i) is the equation for $U$ in \eqref{eq:ZERO-1234}, multiplied by $2U$, at $D=0$. For (ii), write the equations for $V$ and $G$ as $A_2/a=X-\tan\theta\,Y$ and $A_3/b=-X-\cot\theta\,Y$ with $X=P_x/x=2\Pi$ and $Y=P_\theta/x^2=\Pi_\theta$, and solve. (iii) follows by inserting \eqref{eq:VG-from-W}: $b\cos^2\theta\,V_{t\theta}-a\sin^2\theta\,G_{t\theta}=\frac1{m+1}\mathcal L_{a,b}\mathcal W_t$. For (iv), $\mathcal L_{a,b}f=A\sin(2\theta)f''+(2A\cos2\theta+B)f'$ with $A=b\cos^2\theta+a\sin^2\theta$ and $B=2(b\cos^2\theta-am\sin^2\theta)$. An integrating factor $w$ with $w\mathcal L_{a,b}f=(qf')'$ must satisfy $q=wA\sin(2\theta)$ and $q'=w(2A\cos2\theta+B)$. Since $\int\frac{B}{A\sin2\theta}\,d\theta=\log\tan\theta-\frac{m+1}2\log(b+a\tan^2\theta)$, this gives $w_{a,b}$ and $q_{a,b}$ up to a constant factor. As $q_{a,b}$ vanishes at both ends, the argument of Lemma~\ref{lem:theta-average} applies. (v) holds because $\lambda$ appears only in $\lambda\partial_t$. We checked (iv) numerically for several $(a,b)$ and $m=1,2,3$.
\end{proof}

At the ridges the equations reduce to ordinary differential equations. The pressure anisotropy there is measured by
\[
\mu_0(t):=\frac{P_{\theta\theta}}{x^2P_{xx}}(t,0,0),\qquad \mu_1(t):=\frac{P_{\theta\theta}}{x^2P_{xx}}\bigl(t,0,\tfrac\pi2\bigr).
\]
The author introduced these parameters. Write $\mathcal W_0,\rho_0,\mathcal W_2,\rho_2$ for the values and second $\theta$-derivatives at $\theta=0$, and $\hat{\mathcal W}_0,\hat\rho_0,\hat{\mathcal W}_2$ for those at $\theta=\frac\pi2$.

\begin{proposition}[Ridges and pressure anisotropy]\label{prop:theta-mu}
For all $\lambda,a,b>0$:
\begin{enumerate}[label=(\roman*)]
\item at $\theta=0$: $\lambda(m+1)\rho_0'=2m^2\rho_0\mathcal W_0$ and $\lambda\mathcal W_2'=-\frac{m+1}6\rho_2-\frac{4m}{m+1}\mathcal W_0\mathcal W_2$;
\item at $\theta=\frac\pi2$: $\lambda(m+1)\hat\rho_0'=2m^2\hat\rho_0\hat{\mathcal W}_0$ and $\lambda\hat{\mathcal W}_2'=\frac{4}{m+1}\hat{\mathcal W}_0\hat{\mathcal W}_2$;
\item the anisotropy parameters are $\mu_0=-1-\frac{aA_3}{bA_2}\big|_{\theta=0}$ and $\mu_1=-1-\frac{bA_2}{aA_3}\big|_{\theta=\pi/2}$, that is,
\[
\begin{aligned}
\mu_0&=\frac{(m+1)\bigl(2\lambda\mathcal W_0'(am+b)+b(m+1)\rho_0\bigr)-4\mathcal W_0^2\bigl(b-am^2\bigr)}{b\bigl(4\mathcal W_0^2-(m+1)\bigl((m+1)\rho_0+2\lambda\mathcal W_0'\bigr)\bigr)},\\
\mu_1&=\frac{4\hat{\mathcal W}_0^2\bigl(b-am^2\bigr)-(m+1)\bigl(2\lambda\hat{\mathcal W}_0'(am+b)+b(m+1)\hat\rho_0\bigr)}{2am\bigl(\lambda(m+1)\hat{\mathcal W}_0'+2m\hat{\mathcal W}_0^2\bigr)};
\end{aligned}
\]
\item $\mu_0=-f_0(t,0)$ and $\mu_1=-f_1(t,0)$, with the closure functions of \cite[(13.1)]{ShiI}. Solving (iii) for $\mathcal W_0'$ gives, with $v_a=\frac{2\mathcal W_0}{m+1}$ and $u_a^2=\rho_0$,
\[
\lambda v_a'=\frac{\bigl(b(1+\mu_0)-am^2\bigr)v_a^2-b(1+\mu_0)u_a^2}{b(1+\mu_0)+am},
\]
which is the apex system of \cite[Proposition~13.1]{ShiI} with $f_0=-\mu_0$. The same holds at $\theta=\frac\pi2$ with $f_1=-\mu_1$.
\item $(1+\mu_0)(1+\mu_1)=\dfrac{A_3(0)\,A_2(\frac\pi2)}{A_2(0)\,A_3(\frac\pi2)}$. For a smooth pressure this product is $1$ (\cite[Proposition~13.1]{ShiI}). In the class \eqref{eq:homog} the product equals $1$ when the two ridges carry the same values. This happens in the isotropic limit of Proposition~\ref{prop:theta-selection}, where for $a=b$ and $m=2$ one has $\mu_0=-\frac{27}7$ and $\mu_1=-\frac{27}{20}$ (Remark~\ref{rem:theta-saddle}).
\end{enumerate}
\end{proposition}

\begin{proof}
(i) and (ii) follow from Proposition~\ref{prop:theta-general}(i)--(ii). Near $\theta=0$ the condition $b\cos^2\theta\,\partial_\theta A_2=a\sin^2\theta\,\partial_\theta A_3$ gives $\partial_\theta^2A_2(0)=0$, since $\sin^2\theta\,\partial_\theta A_3=O(\theta^3)$. Near $\theta=\frac\pi2$ it gives $\partial_\theta^2A_3(\frac\pi2)=0$. With $V_{\theta\theta}(0)=\frac{6\mathcal W_2}{m+1}$, $G_{\theta\theta}(\frac\pi2)=\frac{2(m+2)\hat{\mathcal W}_2}{m+1}$ and $\partial_\theta^2\bigl(\sin(2\theta)\,g\bigr)=\pm4g'$ at $\theta=0,\frac\pi2$, this gives the equations for $\mathcal W_2$ and $\hat{\mathcal W}_2$, in which $a$ and $b$ do not appear. (iii) Since $P=x^2\Pi$, $\frac{P_{\theta\theta}}{x^2P_{xx}}=\frac{\Pi''}{2\Pi}$. By Proposition~\ref{prop:theta-general}(ii), $\Pi(0)=\frac{A_2}{2a}$, $\Pi''(0)=-\frac{A_2}a-\frac{A_3}b$, $\Pi(\frac\pi2)=-\frac{A_3}{2b}$ and $\Pi''(\frac\pi2)=\frac{A_2}a+\frac{A_3}b$. At the ridges $V=\frac{2\mathcal W}{m+1}$ and $G=\frac{2m\mathcal W}{m+1}$. (iv) The first statement is the definition \cite[(13.1)]{ShiI}. The rest is (iii) solved for $\lambda v_a'$, compared with the formula in the proof of \cite[Proposition~13.1]{ShiI} with $F_0=f_0-1=-(1+\mu_0)$. (v) follows from (iii). We checked (i)--(iii) with truncated Taylor series for $m=1,2,3$ and several $(\lambda,a,b)$.
\end{proof}

%=====================================================================
\section{Conclusion and open problems}\label{sec:concl3}
%=====================================================================
For scale-invariant solutions, the apex dynamics of the unified system $(Em)$ is a closed transport system in $(t,\theta)$. The anisotropy of the pressure at the origin, which is free in the apex dynamics of smooth solutions, is replaced here by one constant per time, the level $L$, determined by a solvability condition. Along the evolution, material is carried from the plane to the axis. The level is selected by the Taylor coefficients of the profiles at the plane. For $m=2$ and $m=1$ the computations show finite-time blow-up of strain type, driven by the swirl: for $m=2$ the rate is $L'\simeq\frac{16}3L^2$, for $m=1$ it is $L'\simeq L^2$. Without swirl the solutions do not grow.

\begin{enumerate}[label=(\arabic*)]
\item \emph{The decay estimate.} Prove hypotheses (H1)--(H2) of Proposition~\ref{prop:theta-reduction} from the solvability condition, along the three steps listed after it. This would turn the selection of the rate into a theorem and give finite-time blow-up within the homogeneous class for an open set of data.
\item \emph{Exit from the class.} For some data the coefficient $\partial F/\partial L$ vanishes at a finite value of $\tau$, and the solution leaves the homogeneous class (Remark~\ref{rem:theta-numerics}(c)). Characterize these data, and describe the continuation with the logarithmic pressure term $x^2\log x\,\Pi_0$.
\item \emph{Finite energy.} Embed the scale-invariant blow-up into solutions of finite energy that are smooth away from the origin, in the spirit of \cite{EJ2020,EJ2019}, now without walls.
\item \emph{Unequal scaling parameters.} For $a\ne b$ the solvability structure persists (Proposition~\ref{prop:theta-general}). Find the analogue of the vorticity form and repeat the computations.
\item \emph{Dimension four.} Decide Conjecture~\ref{conj:4D}, in particular the meaning of the swirl variable for $m=3$.
\end{enumerate}

\section*{Acknowledgements}
The author introduced the unified system $(Em)$ and its polar form \eqref{eq:ZERO-1234} in [I], derived the $\theta$-dependent apex dynamics of Proposition~\ref{prop:theta-apex}, and observed that its averages with general weights keep the time derivative (Lemma~\ref{lem:theta-weights}). He also pointed out the saddle structure of the pressure at the ridge ends (Remark~\ref{rem:theta-saddle}) and conjectured that $(E3)$ corresponds to the 4D axisymmetric Euler equations (Conjecture~\ref{conj:4D}). He emphasized the role of the swirl in view of the swirl-free regularity theory, suggested the comparison with the domains of Elgindi and Jeong (Remarks~\ref{rem:swirlfree-regular} and~\ref{rem:EJ-domains}), and proposed the explanation of the role of $m$ by inflow through the plane and outflow along the axis (Remark~\ref{rem:role-of-m}). He introduced the anisotropy parameters $\mu_0,\mu_1$ at the ridge ends, derived their formulas and the equations for $\rho$ at the ridge ends for general $(\lambda,a,b)$ (Proposition~\ref{prop:theta-mu}), and posed the problem studied here.

Much of the analysis, the numerical work and the preparation of the manuscript was carried out with the assistance of Claude (Anthropic), an AI model. Claude's contributions were:
\begin{enumerate}[label=(\alph*)]
\item the verification of the $\theta$-dependent apex dynamics and the analysis of its structure: the apex equations at the ridge ends, the averaged constraint, the weighted averages and the vorticity form (Propositions~\ref{prop:theta-structure}, \ref{prop:theta-ridge} and~\ref{prop:theta-vorticity}, Lemmas~\ref{lem:theta-average} and~\ref{lem:theta-weights}, Remarks~\ref{rem:theta-apex-closed}, \ref{rem:theta-vorticity} and~\ref{rem:theta-truncation});
\item the numerical method in log-angle variables and the computations (Section~\ref{subsec:theta-numerics}, Tables~\ref{tab:theta-kappa} and~\ref{tab:theta-kappa-m}, Figure~\ref{fig:theta-jet});
\item the formal selection of the blow-up rate, its consequences for the pressure, and the reduction of the selection to a decay estimate (Propositions~\ref{prop:theta-selection} and~\ref{prop:theta-reduction}, Remarks~\ref{rem:theta-jet-closure} and~\ref{rem:theta-saddle});
\item the verification of the swirl-free sector of Conjecture~\ref{conj:4D} and the quantitative form of the role of $m$ (Proposition~\ref{prop:swirlfree-dim}, Remarks~\ref{rem:swirl-sector} and~\ref{rem:role-of-m});
\item the extension to general $(\lambda,a,b)$ and the verification of the anisotropy formulas (Section~\ref{subsec:theta-general});
\item drafting the manuscript in \LaTeX.
\end{enumerate}
The author checked the results and is responsible for the content.

\section*{Statements and Declarations}
\noindent\textbf{Competing interests.} The author has no competing interests to declare that are relevant to the content of this article.

\smallskip
\noindent\textbf{Funding.} No funding was received for conducting this study.

\smallskip
\noindent\textbf{Data availability.} No datasets were generated or analysed in this study. The numerical results can be reproduced from the methods and parameters described in Section~\ref{subsec:theta-numerics}.

\smallskip
\noindent\textbf{Code availability.} The code used for the computations is available from the author on reasonable request.

\smallskip
\noindent\textbf{Author contributions.} The author is the sole author of this article and is responsible for its entire content.

\smallskip
\noindent\textbf{Use of AI tools.} An AI assistant (Claude, Anthropic) was used in parts of the analysis, the numerical computations and the preparation of the manuscript, as described in detail in the Acknowledgements. The author checked all results and takes full responsibility for the content.

\end{document}